\documentclass[aps,prx,twocolumn,superscriptaddress,showpacs,final,floatfix,longbibliography]{revtex4-2}
 
\usepackage[utf8]{inputenc}
\usepackage{amsmath,amssymb,amsfonts}
\usepackage{graphicx}
\usepackage{hyperref}
\usepackage{mathtools}
\usepackage{braket}
\usepackage{xcolor}
\usepackage{cleveref}
\hypersetup{hidelinks}

\newcommand{\Tr}{\mathrm{Tr}}
\newcommand{\MSE}{{\rm MSE}}

\newcommand{\dd}{{\rm d}}

\newcommand{\figplaceholder}[1]{\fbox{\begin{minipage}[c][3.2cm][c]{0.92\linewidth}\centering #1\end{minipage}}}
\newcommand{\maybegraphics}[2][]{\IfFileExists{#2}{\includegraphics[#1]{#2}}{\figplaceholder{Missing figure file: #2}}}
 
\newtheorem{theorem}{Theorem}
\newtheorem{corollary}{Corollary}

\newtheorem{lemma}{Lemma}
 
\usepackage{cleveref}
\crefname{equation}{Eq.}{Eqs.}
\crefname{figure}{Fig.}{Figs.}
\crefname{observation}{Obs.}{Obs.}
\crefname{corollary}{Corollary}{Corollaries}
\crefname{lemma}{Lemma}{Lemmata}

\newcommand{\proof}{\par\noindent\textbf{Proof. }}
\newcommand{\qed}{\hfill$\square$}
 
\begin{document}
 
\title{Finite-Shot Sensitivity for Moment Estimation in Quantum Metrology}
 
\author{Shaowei Du}
\affiliation{State Key Laboratory of Artificial Microstructure and Mesoscopic Physics, School of Physics, Frontiers Science Center for Nano-optoelectronics, Peking University, Beijing 100871, China}
\author{Shuheng Liu}
\affiliation{State Key Laboratory of Artificial Microstructure and Mesoscopic Physics, School of Physics, Frontiers Science Center for Nano-optoelectronics, Peking University, Beijing 100871, China}
\author{Weidong Li}
\affiliation{Shenzhen Key Laboratory of Ultraintense Laser and Advanced Material Technology, Center for Intense Laser Application Technology, and College of Engineering Physics, Shenzhen Technology University, Shenzhen 518118, China}
\author{Luca Pezz\`e}
\affiliation{QSTAR, INO-CNR, and LENS, Largo Enrico Fermi 2, 50125 Firenze, Italy}
\author{Augusto Smerzi}
\email{augustosmerzi@sztu.edu.cn}
\affiliation{Shenzhen Key Laboratory of Ultraintense Laser and Advanced Material Technology, Center for Intense Laser Application Technology, and College of Engineering Physics, Shenzhen Technology University, Shenzhen 518118, China}
\author{Qiongyi He}
\email{qiongyihe@pku.edu.cn}
\affiliation{State Key Laboratory of Artificial Microstructure and Mesoscopic Physics, School of Physics, Frontiers Science Center for Nano-optoelectronics, Peking University, Beijing 100871, China}
\affiliation{Collaborative Innovation Center of Extreme Optics, Shanxi University, Taiyuan, Shanxi 030006, China}
\affiliation{Hefei National Laboratory, Hefei 230088, China}
 
\begin{abstract}
The quantum Cram\'er--Rao bound can be saturated only asymptotically and does not specify how many measurements are needed for a concrete estimator to approach it. We develop a finite-measurement theory for method-of-moments estimation, where the parameter is inferred from the sample mean of a calibrating observable rather than from the full likelihood. For general quantum statistical models, the expansion is written in terms of the calibration curve and the central moments of the measured observable. Nonlinear calibration curves make the usual moment estimator biased at finite measurement number; we construct a bias-corrected estimator with bias \(O(\nu^{-3})\). This gives sensitivity corrections beyond the leading error-propagation term of the chosen moment protocol. We identify a general density-matrix condition under which the full \(1/\nu^2\) correction vanishes. In unitary examples, the leading residual correction appears at order \(1/\nu^3\), is governed by calibration curvature, and can be reduced or cancelled by higher-rank components of the same measured observable. The resulting thresholds quantify how many measurements are needed before the asymptotic sensitivity of a moment-estimation protocol is operationally visible.
\end{abstract}
 
\maketitle
 
\textit{Introduction---}
Precision measurements play a central role in modern physics~\cite{ye2024}, with applications ranging from gravitational-wave detection~\cite{aasi2013,tse2019} and atomic clocks~\cite{pedrozo2020,ye2024} to magnetometry~\cite{wasilewski2010,wolfgramm2010} and bioimaging~\cite{taylor2016,ahmed2026ecore}. Quantum sensing and metrology provide the framework for exploiting quantum states and measurements to estimate physical parameters with sensitivities beyond classical limits~\cite{paris2009,giovannetti2011,degen2017,pezze2018}.
 
The field is usually organized around asymptotic local benchmarks. For a quantum statistical model \(\rho_\theta\), the quantum Cram\'er--Rao bound (QCRB) gives the best leading precision attainable by locally unbiased estimators after optimizing over measurements~\cite{helstrom1969,holevo2011,braunstein1994}. This benchmark underlies quantum-enhanced sensing and the use of quantum Fisher information as a witness of metrologically useful multipartite entanglement and squeezing~\cite{pezze2009,hyllus2012,toth2012}, and continues to guide recent schemes based on nonlinear dynamics and optimized metrological protocols~\cite{imai2025,du2022,zhao2020,kong2026}. A complementary question, which is central in practice, is how a concrete measurement-and-estimator scheme approaches the QCRB when only a finite number of measurements is available.

Reaching the leading bound requires both suitable measurements and appropriate estimator conditions~\cite{pezze2014,gessner2023}. Although optimal measurements are known to exist for single-parameter estimation in ideal settings~\cite{liu2020,kurdzialek2023a,zhou2020}, their implementation may be limited by experimental imperfections such as measurement noise~\cite{datta2011,len2022,zhou2023}. Moreover, standard strategies including maximum-likelihood~\cite{polino2019,hervas2025,hervas2025higher}, method-of-moments~\cite{gessner2019,gessner2020}, and Bayesian estimation~\cite{li2018entropy,rubio2019,han2024bayesian} generally approach their asymptotic benchmarks only in the large-sample limit.
Classical statistics already provides stronger finite-sample benchmarks beyond the Cram\'er--Rao bound~\cite{barankin1949,abel2002,gebhart2024}, and quantum extensions have led to hierarchies of few-measurement bounds and to global-estimation approaches in both frequentist and Bayesian settings~\cite{gessner2023,shi2026global,mukhopadhyay2025global}.

 Recent higher-order likelihood theory, built on classical asymptotic statistics~\cite{rao1965,braunstein1992}, has shown that measurements and states that are equivalent at the QCRB level can differ at finite sample size~\cite{hervas2025,hervas2025higher}. Such likelihood-based methods, however, typically require calibration of the full outcome distribution as a function of the unknown parameter~\cite{pezze2014}. Many quantum-metrology experiments instead use only a few low-order moments of a measured signal, inferring the parameter from the sample mean of a calibrated observable without reconstructing the full likelihood~\cite{gessner2019,gessner2020,fadel2025}. Given such a moment estimator, how many repetitions are needed before the asymptotic sensitivity is operationally reached? The answer depends on finite-measurement corrections that are invisible at leading order.
 
We formulate this finite-measurement problem directly at the level of measured moments. First, for arbitrary quantum statistical models \(\rho_\theta\), we derive the bias-corrected method-of-moments expansion in terms of the calibration curve \(f(\theta)=\Tr[\rho_\theta M]\) and the central moments \(\mu_n(\theta)=\Tr[\rho_\theta(M-f(\theta)\mathbb I)^n]\). Second, we identify the general condition under which the \(1/\nu^2\) correction to the bias-corrected mean-square error vanishes. Third, we illustrate the result in effective qubits, qudits with higher-rank measurement components, and continuous-variable models with linear calibration. The qudit case shows explicitly how measurement components that are invisible at leading order can reduce, and in suitable cases cancel, the leading residual \(1/\nu^3\) correction.

\textit{Moment estimation---}
We consider a real parameter encoded in a density matrix \(\rho_\theta\). In the method of moments (MoM) estimation protocol, a Hermitian observable \(M\) is measured independently \(\nu\) times, producing outcomes \(M_1,\ldots,M_\nu\) and sample mean $\bar M=\frac{1}{\nu}\sum_{k=1}^{\nu}M_k$. Provided the calibration curve $f(\theta)=\Tr[\rho_\theta M]$ of nonzero first-order derivative, the usual MoM estimator is
\begin{equation}
\hat \theta_{0}:=f^{-1}(\bar M)\equiv g(\bar M).
\end{equation}
We define the single-shot central moments $\mu_n(\theta)=\Tr\!\left[\rho_\theta\bigl(M-f(\theta)\mathbb I\bigr)^n\right]$. Expanding \(g(\bar M)\) around the calibrated mean \(\Tr[\rho_\theta M]\) gives the leading error-propagation formula of the mean square error (MSE)
\begin{equation}
 \MSE(\hat\theta_0)
 =
 \frac{1}{\nu}\frac{\mu_2}{(f')^2}
 +O(\nu^{-2}).
\label{eq:leading_error_propagation}
\end{equation}
The above coefficient is the leading sensitivity achieved by measuring the chosen observable \(M\). Throughout the paper, the primes are used to denote derivatives with respect to \(\theta\).
 
At finite \(\nu\), this estimator is generally biased. Let \(\delta\bar M=\bar M-f(\theta)\) and expand the inverse calibration function around \(f(\theta)\):
\begin{equation}
 g(f+\delta\bar M)
 =
 \theta+\sum_{r\ge1}\frac{g^{(r)}(f)}{r!}\,\delta\bar M^r .
\label{eq:g_expansion}
\end{equation}
Here \(g^{(r)}\) denotes the \(r\)-th derivative of \(g\) with respect to its argument, evaluated at \(f(\theta)\). Using the mapping of central moments from single-shot to i.i.d. mean, we formulate the moments of the sample mean fluctuation, i.e. \(E[\delta\bar M^n]\), into the powers of \(1/\nu\) and obtains a series expansion of the bias of \(\hat \theta_0\) with respect to the true value
\begin{equation}
 \mathbb E_\theta[\hat\theta_0]-\theta
 =
 \frac{b_1(\theta)}{\nu}
 +\frac{b_2(\theta)}{\nu^2}
 +O(\nu^{-3}),
\label{eq:bias_expansion}
\end{equation}
with \( b_1(\theta)=\frac{1}{2}g^{(2)}(f)\mu_2,\) and \( b_2(\theta)=\frac{1}{6}g^{(3)}(f)\mu_3+\frac{1}{8}g^{(4)}(f)\mu_2^2 .\)
See SM~\cref{sm:framework} for details. 

As demonstrated in Eq.~\eqref{eq:bias_expansion}, the usual MoM estimator $\hat \theta_0$ is asymptotically consistent, but a nonlinear calibration curve produces a finite-measurement bias. Consequently, the MSE of MoM estimator is corrected in this regime in the form of \(1/\nu\)-series
\begin{equation}\label{eq:general_unc_mse}
 \MSE(\hat \theta_0)=
 \frac{A}{\nu}+\frac{V_0}{\nu^2}+O(\nu^{-3}),
\end{equation}where we call \(A:= {\mu_2}/{(f')^2}\), \(V_0:=-{\mu_3 f''}/{(f')^4}+{\mu_2^2}\left(15(f'')^2/4- f'f'''\right)/{(f')^6}\). As seen from Eq.~\eqref{eq:general_unc_mse}, the finite-\(\nu\) bias does not change the leading \(1/\nu\) sensitivity, but it should be removed to be locally unbiased for comparison with Cram\'er--Rao-type benchmarks beyond leading order. See SM~\cref{sm:framework} for details.
 
\textit{Bias-corrected estimation protocol---}
Consequently, we introduce the bias-corrected moment estimator
\begin{equation}
 \hat\theta_{\rm bc}:=g(\bar M)-\frac{b_1[g(\bar M)]}{\nu}-\frac{\widetilde b_2[g(\bar M)]}{\nu^2},
\label{eq:bias_corrected_estimator}
\end{equation}
where \(\widetilde b_2=b_2-b_1 b_1'-Ab_1''/2\). \(b_1\) and \(\widetilde b_2\) are assumed to be known calibration functions determined by the model and the measured observable. Thanks to the additional terms, the bias of \(\hat \theta_{\rm bc}\) is suppressed to \(O(1/\nu^3)\), i.e. \(\mathbb E_\theta[\hat\theta_{\rm bc}]=\theta+O(\nu^{-3})\). Using the same expanding techniques as Eq.~\eqref{eq:general_unc_mse}, one obtains
\begin{equation}
 \MSE(\hat\theta_{\rm bc})= \frac{A}{\nu}+\frac{B_M}{\nu^2}+O(\nu^{-3}),
\label{eq:general_bc_mse}
\end{equation}
where
\begin{equation}
 B_M=V_0-2Ab_1'=-\frac{\mu_3 f''}{(f')^4} + \frac{\mu_2\mu_2'f''}{(f')^5} +\frac{\mu_2^2(f'')^2}{2(f')^6}.
\label{eq:general_BM}
\end{equation}
Here all quantities are evaluated at the operating point. The cancellation of the \(f'''\) terms is a direct consequence of correcting the finite-\(\nu\) bias. Eqs.~\eqref{eq:general_bc_mse} and~\eqref{eq:general_BM} are therefore finite-measurement corrections for the bias-corrected MoM estimator, expressed only in terms of the calibrated mean and the central moments of the measured observable. See SM~\cref{sm:bias_correction} for details. 
 
The coefficient \(B_M\) need not vanish in a generic protocol. When it is nonzero, Eq.~\eqref{eq:general_bc_mse} simply describes the \(O(1/\nu)\) relative approach to the asymptotic moment sensitivity. However, Eq.~\eqref{eq:general_BM} also identifies protocols for which the full \(1/\nu^2\) correction is cancelled:
\begin{equation}
 B_M=0.
\label{eq:BM_zero_condition}
\end{equation}
This is a finite-measurement optimization criterion, which involves only the calibrated mean \(f(\theta)=\Tr[\rho_\theta M]\), its derivatives, and the central moments of the measured observable. It is independent of purity and does not require unitary dynamics. When Eq.~\eqref{eq:BM_zero_condition} holds, the leading finite-measurement correction is postponed to order \(1/\nu^3\). In practice, Eq.~\eqref{eq:BM_zero_condition} can be used either as a diagnostic of a given moment protocol or as an optimization condition over experimentally accessible observables.

\textit{Achieved sensitivity and optimal benchmarks---}
The expansion above gives the sensitivity achieved by the chosen bias-corrected moment estimator. At leading order,
\begin{equation}
 \frac{\mu_2}{(f')^2}
 \ge
 \frac{1}{F_C[M]}
 \ge
 \frac{1}{F_Q},
\label{eq:moment_leading_bound}
\end{equation}
where \(F_C[M]\) is the classical Fisher information of the outcome distribution by measuring \(M\), and \(F_Q\) is the quantum Fisher information (QFI). Thus the MoM estimator generally achieves a sensitivity above the Cram\'er--Rao benchmarks.  See SM~\cref{sm:density_matrix_benchmark} for the proof.
 
Generally, the QFI is given by \(F_Q=\Tr (\rho_\theta L_\theta^2)\), where \(L_\theta\) is the symmetric logarithmic derivative (SLD) defined by
\begin{equation}\label{SLD_definition}
\partial _\theta \rho _\theta=\frac12(\rho_\theta L_\theta+L_\theta\rho_\theta) ,
\end{equation}
representing the first-order differential feature of the state \(\rho_\theta\) . Equality with \(1/F_Q\) is a special case, obtained only by first-order optimal observables (FOOs). We prove that in this regime the SLD, up to an irrelevant shift and scale, is a sufficient condition of QCRB saturation. Considering the pure unitary schemes, Eq.~\eqref{SLD_definition} is also necessary in the sense that its solutions characterize all FOOs. See SM~\cref{sm:optimal_measurement} for the proof.  Beyond leading order, no coefficient-by-coefficient ordering is implied unless a corresponding higher-order bound is derived under the same bias and regularity assumptions.

\textit{Example I: effective qubits---}
We now specialize the general density-matrix result to pure unitary families, where closed analytic corrections can be obtained. Consider \(|\psi_\theta\rangle=e^{-i\theta H}|\psi\rangle\) in an effective 2-dimensional moment subspace. By referring to as ``qubits", we include not only the 2-dimensional system but also the schemes that only two dimensions are engaged in the process. It is common in practice since cases that attain the Heisenberg limit all belong to this scenario. Let us denote the central moments of the Hamiltonian with \(m_n=\langle(H-\langle H\rangle)^n\rangle\). In this case, we have \(F_Q=4m_2\). In a local basis adapted to the state and its tangent direction, from Eq.~\eqref{SLD_definition}, an generic FOO can be written as
\begin{equation}
 M=\begin{pmatrix}0&1\\1&\lambda\end{pmatrix},\lambda\in \mathbb R.
\label{eq:qubit_observable}
\end{equation}
For the bias-corrected estimator, we have simply \(B_M={\lambda^2}/{(32m_2})\). The complete \(1/\nu^2\) correction can be directly cancelled by a centered choice of \(\lambda=0\) and one obtains the first residual term from Eq.~\eqref{eq:general_bc_mse} as
\begin{equation}
 \MSE(\hat\theta_{\rm bc})
 =
 \frac{1}{4m_2\nu}
 +\frac{(4m_2^3+m_3^2)^2}{384m_2^7\nu^3}
 +O(\nu^{-4}).
\label{eq:qubit_third_order_bc}
\end{equation}
Thus the first correction to the leading sensitivity is suppressed by \(1/\nu^2\) relative to the leading term. See SM~\cref{sm:qubit} for details and a comparison with $\hat\theta_0$ and~\cref{fig:qubit_bias_corrected} for illustration of the finite-shot saturation of QCRB with qubit model.
 
As an important example, we consider the generic scenario that reaches Heisenberg limit. Given the Hamiltonian, the probe state is cat-like state superposed by the extremal eigenstates corresponding to minimal and maximal eigenvalues, for example, $H=J_z$ with the NOON state $\frac{1}{\sqrt2}(\ket{N,0}+\ket{0,N})$ in $z$ direction. Such Heisenberg-limited probes have a null \(m_3\) and Eq.~\eqref{eq:qubit_third_order_bc} reduces to
\begin{equation}
 \MSE(\hat\theta_{\rm bc})
 =
 \frac{1}{4m_2\nu}
 +\frac{1}{24m_2\nu^3}
 +O(\nu^{-4}).
\label{eq:HL_third_order}
\end{equation}
Despite different models, they display essentially uniform convergence behaviors.
 
\begin{figure}[t]
\centering
\maybegraphics[width=0.95\columnwidth]{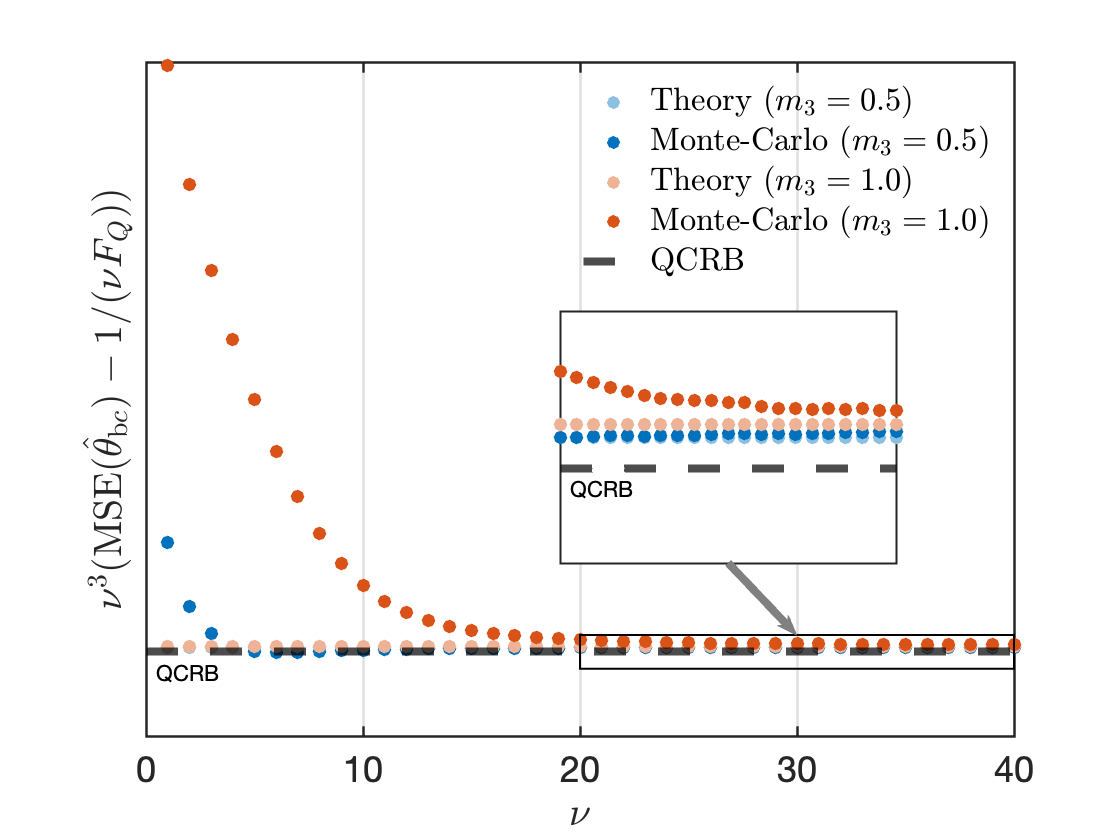}
\caption{(I) Qubit Monte-Carlo benchmark from the bias-corrected analysis. We provide two example with the same \(F_Q\) but different \(m_3\) in the unit of \(m_2\) and investigate their finite-shot sensitivities with increasing sample size. As predicted by Eq.\eqref{eq:qubit_third_order_bc}, the Monte-Carlo results of the bias-corrected estimator \(\hat \theta_{\rm bc}\) display cubit-law convergence with the QCRB beyond the leading order. Qubit probes with different \(m_3\) demonstrate different saturating behaviors according as well to Eq.\eqref{eq:qubit_third_order_bc}.}
\label{fig:qubit_bias_corrected}
\end{figure}
 
\textit{Example II: qudits and higher-rank measurements---}
In effective dimension larger than two, first-order optimality does not fix all components of the measured observable. Besides the component that fixes the leading sensitivity, the observable may contain higher-rank components that are invisible at leading order but enter higher derivatives of the calibrated moment. See SM~\cref{sm:optimal_measurement} for details. 
 
In this regime, we denote by \(\alpha\) the relevant component coupling the first tangent direction to the next moment direction, which efficiently identifies the FOOs. Define \( G_Q^{(2)} =m_4-m_2^2-m_3^2/m_2\), after the centered choice that cancels the \(1/\nu^2\) term, the leading residual correction becomes
\begin{equation}
 \MSE(\hat\theta_{\rm bc})
 =
 \frac{1}{4m_2\nu}
 +\frac{D_\alpha^2}{384m_2^5\nu^3}
 +O(\nu^{-4}),
\label{eq:third_order_main}
\end{equation}
where \(D_\alpha= 3m_2^2+m_4-3\alpha m_2\sqrt{G_Q^{(2)}}\). The coefficient is non-negative and controlled by the residual curvature of the calibrated moment. See SM~\cref{sm:qudit} for details. 

Note that \(G_Q^{(2)}\) is positive for non-trivial qudit probe states so that the choice
\begin{equation}
 \alpha_{\rm opt}= \frac{3m_2^2+m_4}{3m_2\sqrt{G_Q^{(2)}}}
\label{eq:alpha_opt}
\end{equation}
cancels the displayed \(1/\nu^3\) term within the same local perturbative expansion. This cancellation requires sufficient measurement freedom and does not constitute a global unconstrained optimum.
 
A useful solvable benchmark is the qutrit model
\begin{equation}
 H_t=\begin{pmatrix}
 0&-i&0\\ i&0&-i\sqrt2\\0&i\sqrt2&0
 \end{pmatrix},
 \qquad
 M_t=\begin{pmatrix}
 0&1&0\\1&0&\alpha\\0&\alpha&0
 \end{pmatrix}.
\label{eq:qutrit_model}
\end{equation}
Here we choose \(m_2=1\), \(m_3=0\), \(m_4=3\), and \(G_Q^{(2)}=2\). Equation~\eqref{eq:third_order_main} predicts
\begin{equation}
 \MSE(\hat\theta_{\rm bc})
 =
 \frac{1}{4\nu}
 +\frac{3(\alpha^2-2\sqrt2\alpha+2)}{64\nu^3}
 +O(\nu^{-4}),
\label{eq:qutrit_check}
\end{equation}
with cancellation of the displayed \(1/\nu^3\) term at \(\alpha=\sqrt2\). See~\cref{fig:qutrit_bias_corrected} for illustration of the finite-shot saturation of the QCRB with this qutrit model. The expansion is asymptotic for fixed \(\alpha\); for larger \(\alpha\), the cubic regime may require larger \(\nu\) to become visible.
 
\begin{figure}[t]
\centering
\maybegraphics[width=0.95\columnwidth]{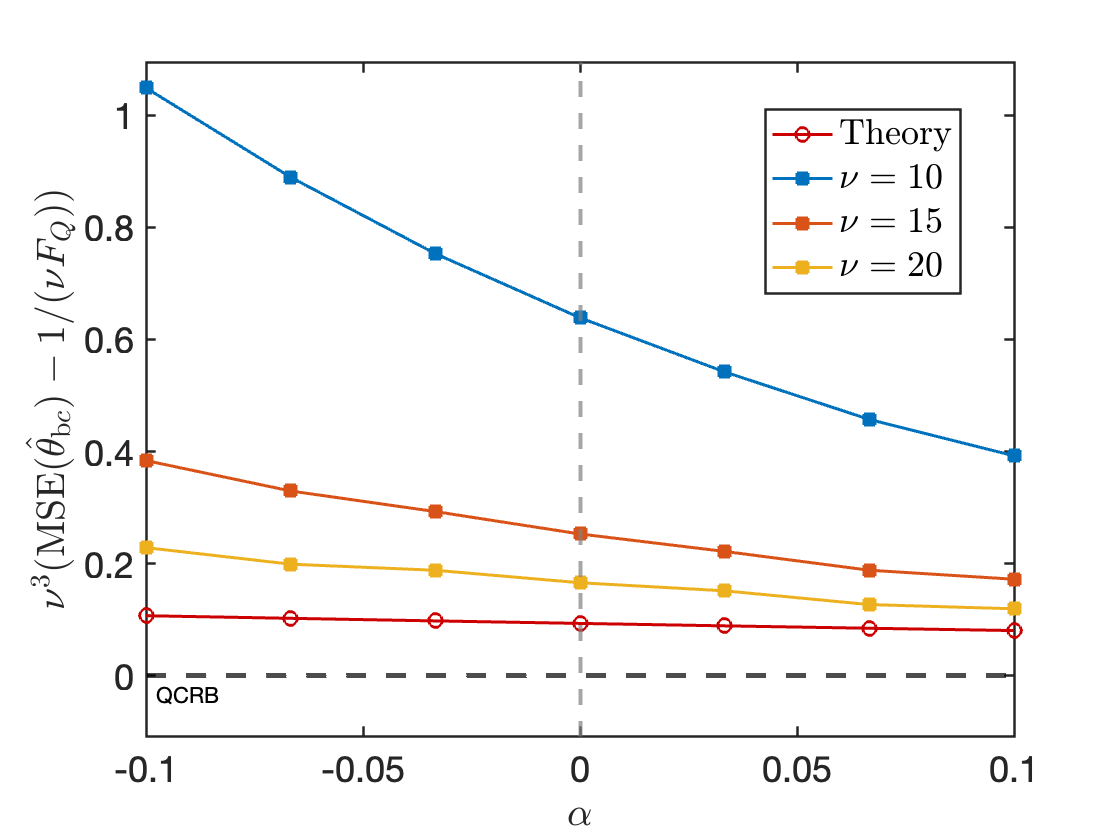}
\caption{(II) Qutrit Monte-Carlo benchmark from the bias-corrected analysis. We consider the solvable model in Eq.~\eqref{eq:qutrit_model} with varying first-order optimal observables indicated by \(\alpha\). Non-zero \(\alpha\) indicates the involvement of higher-rank component, which results in different QCRB saturations. With increasing \(\nu\), the next-to-leading order sensitivities of biased-corrected estimator \(\hat \theta_{\rm bc}\) converge to the theoretical predictions in~Eq.~\eqref{eq:qutrit_check} with a cubic law.}
\label{fig:qutrit_bias_corrected}
\end{figure}
 
\textit{Example III: continuous-variable linear calibration---}
The density-matrix expansion Eq.~\eqref{eq:general_bc_mse} also covers continuous-variable models. A particularly transparent case is an exactly linear calibration curve,
\begin{equation}
 f(\theta)=f_0+f'_0\theta,
 \qquad
 f'_0\neq0.
\label{eq:linear_calibration}
\end{equation}
Then the inverse function \(g\) is also linear and the usual MoM estimator \(\hat \theta_0\) is already unbiased at finite \(\nu\). The inverse-map corrections are absent and
\begin{equation}
 \MSE(\hat\theta_0)=\frac{1}{\nu}\frac{\mu_2}{(f_0')^2}.
\label{eq:linear_exact_mse}
\end{equation}
This includes ideal displacement estimation with a conjugate quadrature measurement. If the measured quadrature is also an FOO for the chosen probe state, for example the squeezed vacuum, Eq.~\eqref{eq:linear_exact_mse} reaches the leading quantum benchmark for every \(\nu\). 
 
\textit{Finite-measurement thresholds---}
The expansion provides a natural threshold for when the asymptotic sensitivity of the chosen moment protocol is operationally reached. If \(B_M\neq0\), we rewrite Eq.~\eqref{eq:general_bc_mse} as
\begin{equation}
 \MSE(\hat\theta_{\rm bc}) = \frac{1}{\nu}\frac{\mu_2}{(f')^2} \left[1+ \frac{B_M(f')^2}{\mu_2}\frac{1}{\nu} +O(\nu^{-2})\right].
\end{equation}
For a relative tolerance \(\varepsilon\), the leading term is accurate when
\begin{equation}
 \nu
 \gtrsim\nu_2^*=\frac{|B_M|(f')^2}{\varepsilon\mu_2}.
\label{eq:threshold_BM}
\end{equation}
If \(B_M=0\), the threshold is instead controlled by the first nonvanishing higher-order correction. The unitary qudit example~Eq.~\eqref{eq:third_order_main} gives the threshold
\begin{equation}
 \nu_3^*= \frac{|D_\alpha|}{\sqrt{96\varepsilon}\,m_2^2}.
\label{eq:threshold_Dalpha}
\end{equation}
Thus finite-shot corrections quantify the number of repetitions required before the asymptotic sensitivity of the chosen moment-estimation protocol is operationally visible. When the leading coefficient saturates the relevant quantum benchmark, the same threshold gives the measurement budget needed to make the corresponding metrological resource experimentally visible.
 
\textit{Discussion---}
The density-matrix expansion Eq.~\eqref{eq:general_bc_mse} shows that the MoM approach is not restricted to pure states or unitary dynamics. Its basic ingredients are the calibrated mean \(f(\theta)\) and the central moments \(\mu_n(\theta)\) of the measured observable. The resulting expansion is not a new lower bound, but the finite-shot sensitivity achieved by the specific bias-corrected estimator.
 
This distinction is important. At leading order, the moment estimator reaches the optimal quantum benchmark only when the measured observable is an FOO. Beyond leading order, the coefficients derived here quantify how the chosen moment protocol approaches its asymptotic sensitivity at finite measurement number. In particular, the general density-matrix expression identifies the condition \(B_M=0\) under which the full \(1/\nu^2\) correction vanishes. When this condition is satisfied, the leading residual correction is postponed to order \(1/\nu^3\).
 
In the unitary qudit models analyzed explicitly, the \(1/\nu^3\) term can be written in closed form as a positive calibration-curvature contribution. In suitable cases, this contribution can be reduced, or cancelled, by optimizing additional components of the measured observable. This suggests an order-by-order strategy for improving finite-shot performance, but cancellations beyond the orders derived here require the corresponding higher-order expansion and are not claimed in this work. A complete optimization should also include experimental constraints on the measurement spectrum, operator norm, and implementation. Multiparameter extensions will require a multivariate bias-corrected inverse map and the appropriate treatment of rectangular moment maps.
 
\textit{Conclusion---}
We developed a density-matrix method-of-moments expansion for quantum metrology with a finite number of measurements. The usual moment estimator is asymptotically consistent but generally biased at finite measurement number when the calibration curve is nonlinear. We constructed a bias-corrected estimator and derived its mean-square error expansion in terms of the calibrated mean and the central moments of the measured observable.
 
The expansion gives the achieved sensitivity of the chosen moment protocol. Its \(1/\nu^2\) correction is governed by a general density-matrix coefficient \(B_M\); protocols satisfying \(B_M=0\) exhibit no \(1/\nu^2\) correction. Effective qubits, qudits, and continuous-variable linear calibration illustrate three distinct regimes: a centered qubit observable removes the second-order term, higher-rank qudit measurements can also reduce the leading \(1/\nu^3\) curvature correction, and exactly linear moment maps have no nonlinear inverse-map correction. The resulting thresholds quantify how many measurements are required before the asymptotic sensitivity of a moment-estimation protocol becomes operationally visible.

Several directions follow naturally from this work. First, the cancellation conditions derived here suggest an order-by-order optimization of finite-measurement performance under realistic constraints on the observable spectrum, norm, and implementation. Second, the different finite-shot behaviors of effective qubits and qudits point to a geometric interpretation in terms of how higher moment directions enter the calibrated signal. Finally, extending the present construction to multiparameter estimation would require a multivariate bias-corrected inverse map and should be compared with finite-sample corrections to multiparameter Cram\'er--Rao and Holevo-type bounds.
 
\begin{acknowledgments}
We thank Hai-Long Shi for discussions. This work is supported by the National Natural Science Foundation of China (Grant Nos. 12125402, 12534016), the Beijing Natural Science Foundation (Grant No. Z240007), and the Quantum Science and Technology - National Science and Technology Major Project (Grant Nos. 2024ZD0302401 and 2021ZD0301500). The work of A.S. is supported by the National Natural Science Foundation of China (Grant No. W2531008) and the Peacock Plan, Quantum Science Strategic Special Project of Guangdong Province (GDZX2505001)  and the Quantum Science and Technology - National Science and Technology Major Project (2025ZD0300800).
\end{acknowledgments}
 
\bibliography{refs}
 
\clearpage
\onecolumngrid
\appendix
\section*{Supplemental Material for  \\ ``Finite-Shot Sensitivity for Moment Estimation in Quantum Metrology''}
\setcounter{section}{0}
\setcounter{figure}{0}
\renewcommand{\thesection}{S\arabic{section}}
\renewcommand{\theequation}{S\arabic{section}.\arabic{equation}}
\renewcommand{\thefigure}{S\arabic{figure}}
 
\setcounter{equation}{0}
\section{Method of moments in the non-asymptotic regime}\label{sm:framework}
The method of moments (MoM) protocol uses the measured mean value $\bar M=\frac1\nu \sum_i M_i$ of some observable $M$ to estimate the parameter $\theta$ of the system in the light of $\langle M\rangle_\theta:=f(\theta)$. Explicitly, we put the analytical form of MoM as\begin{equation}
    \hat \theta_0:=f^{-1}(\bar M)\equiv g(\bar M).
\end{equation}
 
In the spirit of central limit theorem (CLT), when $\nu\rightarrow\infty$, we have $\bar M\rightarrow\langle M\rangle_\theta$. In this regime, by taking into account the first-order biasedness of MoM, i.e. $\hat \theta_{\rm MoM}=\theta + \frac{1}{\frac{\partial \langle M \rangle_\theta }{\partial \theta}}(\bar M-\langle  M \rangle_\theta)+O((\bar M-\langle  M \rangle_\theta)^2)$, one recovers the error propagation formula (EPF)~\cite{pezze2014,gessner2019}\begin{equation}
    \lim_{\nu\rightarrow\infty}\Delta^2\hat \theta_0=\frac{1}{\nu}\frac{(\Delta M)^2_\theta}{(\partial_\theta \langle M\rangle_\theta)^2 },
\end{equation}where $(\Delta M)^2_\theta:=\Tr[\rho_\theta(M-f(\theta )\mathbb I)^2]$.
 
However, in the non-asymptotic regime when $\nu$ is not large enough, which is often the case in reality, the EPF no longer gives an accurate description of the precision: firstly, the CLT does not hold and the non-Gaussianity of the probability distribution function (PDF) of $\bar M$ starts to play a non-negligible role; secondly, the higher-order unbiasedness of $\hat \theta_{\rm MoM}$ should also be considered. Therefore, in order to study the MoM estimation protocol in this regime, we incorporate the strict forms of both the PDF and the estimator to derive a series expansion of $\Delta^2\hat \theta_{\rm MoM}$ w.r.t. $1/\nu$. 
 
We start to define the standardized sum by\begin{equation}
    M_0:=\frac1{\sqrt \nu}\sum _{i=1}^\nu \frac{(M_i-\langle\hat M\rangle_\theta)}{(\Delta\hat M)_\theta}\equiv \frac{\sqrt\nu}{(\Delta\hat M)_\theta}\delta \bar M.
\end{equation}
 
On one hand, the Edgeworth series is adopted to expand the real PDF~\cite{kolassa2006,mccullagh2018}\begin{equation}
    \begin{aligned}
        p(M_0\mid \theta)=\phi(M_0)[&1\\&+\kappa_3h_3(M_0)/6\\&+(\kappa_4 h_4(M_0)/24+10\kappa_3^2h_6(M_0)/720)\\&+(\kappa_5h_5(M_0)/120+35\kappa_3\kappa_4h_7(M_0)/5040+280\kappa_3^3h_9(M_0)/362880)\\&\cdots],
    \end{aligned}
\end{equation}where we denote the standard normal distribution with $\phi(\cdot)$ and the Hermite polynomials $h_n(\cdot)$\begin{equation}
    h_j(x)=(-1)^j[\frac{\dd^j}{\dd x^j}\exp(-x^2/2)]/\exp(-x^2/2),
\end{equation}Note that the cumulants $\kappa_n$ scale with $O(\frac1{n^{r/2-1}})$ when $r\le 3$ and we list the first few\begin{equation}
    \begin{aligned}
        &\kappa_2=1,\\&\kappa_3=\frac{1}{\sqrt\nu\mu_2^{\frac32}}\mu _3,\\&\kappa_4=\frac{1}{\nu\mu_2^{2}}(\mu_4-3\mu_2^2),
    \end{aligned}
\end{equation}where we denote the $n$-th central moment of $M$ as\begin{equation}
    \mu_n(\theta):=\Tr[\rho_\theta(M-f(\theta )\mathbb I)^n].
\end{equation}Therefore, terms of the same order in the Edgeworth series are grouped together. 
 
On the other hand, we can fully expand the estimator with $M_0$ as
\begin{equation}
    g(f+\delta\bar M)
 =
 \theta+\sum_{r\ge1}\frac{g^{(r)}(f)}{r!}\,\delta\bar M^r .
\end{equation}
where we denote $g_n:=g^{(n)}(\langle \hat M \rangle_\theta)$ and the first three are given with the inverse function formulae\begin{equation}
    \begin{aligned}
        &g_1=\frac{1}{f'},\\&g_2=-\frac{f''}{(f')^3},\\&g_3=\frac{3(f'')^2-f'f'''}{(f')^5}.
    \end{aligned}
\end{equation}Note that throughout the text we use primes to denote the derivative w.r.t. $\theta$ unless pointing out specially.
 
With the above preparation, we can compute the moments of $M_0$ in different orders in preparation\begin{equation}
\begin{aligned}
    \mathbb E[M_0]&=0,\qquad \mathbb E[M_0^2]=1,\\\mathbb E[M_0^3]&=\frac{\mu_3}{\mu_2^{\frac32}\sqrt{\nu}},\qquad\mathbb E[M_0^4]=3+\frac{\mu_4-3\mu_2^2}{\nu\mu_2^2},
\end{aligned}
\end{equation}
which lead to
\begin{equation}
\begin{aligned}
\mathbb E[\delta\bar M]&=0,
\qquad
\mathbb E[\delta\bar M^2]=\frac{\mu_2}{\nu},\\
\mathbb E[\delta\bar M^3]&=\frac{\mu_3}{\nu^2},
\qquad
\mathbb E[\delta\bar M^4]=
\frac{3\mu_2^2}{\nu^2}
+\frac{\mu_4-3\mu_2^2}{\nu^3}.
\end{aligned}
\label{eq:sample_mean_moments}
\end{equation}The expectation value of $\hat \theta_0$ and $\hat \theta^2_0$ is then derived respectively
\begin{equation}\label{mom_mean_unc}
    \mathbb E[\hat\theta_0]=\theta+\frac1{\nu}\cdot \frac12g_2\mu_2+\frac1{\nu^2}\cdot(\frac16g_3\mu_3+\frac18g_4\mu_2^2)+O(\frac1{\nu^3}),
\end{equation}
\begin{equation}
    \mathbb E[\hat\theta_0^2]=\theta^2+\frac1\nu (g_1^2+g_2\theta)\mu_2+\frac1{\nu^2}((\frac13g_3\theta+g_1g_2)\mu_3+3\mu_2^2(\frac14 g_2^2+\frac13 g_1g_3+\frac1{12}\theta g_4))+O(\frac1{\nu^3}),
\end{equation}
We can then provide the correction to the variance in the non-asymptotic regime.
\begin{equation}\label{mse2}
    \begin{aligned}
        \MSE(\hat \theta_0)&=\frac1\nu g_1^2\mu_2+\frac1{\nu^2}(g_1g_2\mu_3+\mu_2^2(\frac34 g_2^2+g_1g_3))+O(\frac1{\nu^3})\\&=\frac1\nu \frac{\mu_2}{(f')^2 }+\frac1{\nu^2}(-\mu_3\frac{f''}{(f')^4}+\mu_2^2\frac{\frac{15}{4}(f'')^2-f'f'''}{(f')^6})+O(\frac1{\nu^3}).
    \end{aligned}
\end{equation}Let us denote\begin{equation}
    A:= {\mu_2}/{(f')^2},
\end{equation}\begin{equation}
    V_0:=-\frac{\mu_3 f''}{(f')^4}+\frac{\mu_2^2}{(f')^6}\left(\frac{15}{4}(f'')^2- f'f'''\right).
\end{equation}
We recover the Eq.~\eqref{eq:general_unc_mse} in the main text.
 
In a unitary process $U(\theta)=\exp(-i\theta H)$, we can use the von Neumann equation\begin{equation}\label{vNeq}
    \frac{\partial \langle \hat M \rangle_\theta }{\partial \theta}=-{\mathrm i}\langle[\hat M_\theta,\hat H]\rangle_{\rho},
\end{equation}where the r.h.s. is calculated actually with the initial state $\rho$. As previous, we focus on the vicinity of $\theta=0$ ~\cite{gessner2019,fadel2025}.
 
We can observe directly from~\cref{mse2} that there are special cases in which all the higher-order corrections vanish, which are summarized with the following corollary
\begin{corollary}
    The higher-order corrections are nullified if the measurement is canonically conjugate to the Hamiltonian\begin{equation}
        \MSE(\hat \theta_0)=\frac1\nu \frac{(\Delta H_c)^2}{(\partial_\theta \langle H_c\rangle_\theta)^2 }, \ s.t. [H,H_c]=i.
    \end{equation}
\end{corollary}
\proof  In this scenario, the higher-order derivatives of $\langle H_c\rangle_\theta$ is zero, due to the fact that\begin{equation}
    \partial_\theta^n \langle H_c\rangle_\theta\equiv0 , \forall n\ge 2
\end{equation}which is brought by the canonical commutation relation. Therefore, we have $g_n\equiv 0, \forall n\ge 2$, which automatically cancels out the higher-order terms in~\cref{mse2}.\qed
 
As a widely used example, consider that we have a general single-mode Gaussian state $\ket{r,\alpha}=D(\alpha)S(r)\ket0$. We choose the squeezing along $x$-axis and parametrize it by a corresponding displacement $U(\theta)=\exp(-i\theta p)$ (we have chosen $[x,p]=\frac i2$). The QFI of this scheme reads\begin{equation}
    F_Q(p)=4\Delta^2p=e^{2r},
\end{equation}while the error-propagation formula measuring $x$ is \begin{equation}
    \frac{\Delta^2x}{|\langle[x,p]\rangle|^2}=e^{-2r}=\frac1{F_Q(p)},
\end{equation}so that we notice $x$ is also an optimal observable to satisfy QCRB. It can be easily shown that the higher-order corrections \textit{vanish}\begin{equation}
    \MSE(\hat \theta_0)=\frac1\nu e^{-2r}
\end{equation}
 
\section{Bias correction}\label{sm:bias_correction}
The bias of the usual MoM estimator $\hat \theta_0$ is
\begin{equation}
 \mathbb E[\hat\theta_0]-\theta=\frac{b_1}{\nu}+\frac{b_2}{\nu^2}+O(\nu^{-3}).
\end{equation}
From Eq.~\eqref{mom_mean_unc}, we see that
\begin{equation}
 b_1(\theta)= \frac{1}{2}g_2\mu_2= -\frac{\mu_2 f''}{2(f')^3},
\end{equation}
and
\begin{equation}
 b_2(\theta) =\frac{1}{6}g_3\mu_3+\frac{1}{8}g_4\mu_2^2.
\end{equation}The finite-\(\nu\) bias does not change the leading \(1/\nu\) sensitivity, but it should be removed to be locally unbiased for comparison with Cram\'er--Rao-type benchmarks beyond leading order.
 
By modifying the $\hat \theta_0$ by subtracting the lower-order biased terms, we construct a new  biased-corrected estimator\begin{equation}
\hat\theta_{\rm bc}=g(\bar M)-\frac{b_1[g(\bar M)]}{\nu}-\frac{\widetilde b_2[g(\bar M)]}{\nu^2}.
\end{equation}
The correction functions \(b_1,\tilde b_2\) are chosen so that
\begin{equation}
 \mathbb E_\theta[\hat\theta_{\rm bc}] =\theta+O(\nu^{-3}).
\end{equation}
With a term-by-term comparison, we find
\begin{equation}
 \widetilde b_2
 =b_2-b_1 b_1'-\frac{1}{2}\frac{\mu_2}{(f')^2}b_1'' .
\end{equation}
Note that the biased-corrected estimator \(\hat\theta_{\rm bc}\), just as the original \(\hat \theta_0\), is solely determined by the physical model and the measurement itself. We complicate the calibration curve to make it \(\nu\)-dependent so that the bias up to \(1/\nu^2\) is cancelled. This also leads to a suppression of error in the finite \(\nu\) regime.
 
We use the same techniques in \cref{sm:framework} but preserve the \(1/\nu^3\) terms. The expectation value of $\hat \theta_0$ and $\hat \theta^2_0$ is then derived respectively
\begin{equation}
   \mathbb E[\hat\theta_0]=\theta+\frac1{\nu}\cdot \frac12g_2\mu_2+\frac1{\nu^2}\cdot(\frac16g_3\mu_3+\frac18g_4\mu_2^2)+\frac1{\nu^3}\cdot(\frac1{24}g_4(\mu_4-3\mu_2^2)+\frac1{12}g_5\mu_2\mu_3+\frac{1}{48}g_6\mu_2^3)+O(\frac1{\nu^4}),
\end{equation}
\begin{equation}
    \begin{aligned}
        \mathbb E[\hat\theta_0^2]&=\theta^2+\frac1\nu (g_1^2+g_2\theta)\mu_2+\frac1{\nu^2}((\frac13g_3\theta+g_1g_2)\mu_3+3\mu_2^2(\frac14 g_2^2+\frac13 g_1g_3+\frac1{12}\theta g_4))\\&+\frac1{\nu^3}((\mu_4-3\mu_2^2)(\frac14 g_2^2+\frac13 g_1g_3+\frac1{12}\theta g_4)+10\mu_2\mu_3(\frac1{60}g_5\theta+\frac{1}{12}g_1g_4+\frac16g_2g_3)\\&+15\mu_2^3(\frac1{360}g_6\theta+\frac1{60}g_1g_5+\frac1{24}g_2g_4+\frac1{36}g_3^2)+O(\frac1{\nu^4}).
    \end{aligned}
\end{equation}
Thus one obtains
\begin{equation}\label{third_explicit_bc}
 \MSE(\hat\theta_{\rm bc}) =\frac{A}{\nu}+\frac{B_M}{\nu^2}+\frac{D_M}{\nu^3}+O(\nu^{-4}),
\end{equation}
where
\begin{equation}
 B_M=V_0-2A b_1',
\end{equation}
 
\begin{equation}\label{third_order_bc_term}
\begin{aligned}
D_M
=\frac{1}{12(f')^{10}}
\Bigl(
&4 f''' f'^5 \left(3\mu_2^2-\mu_4\right)
\\
&+f'^4
\Bigl[
15f''^2(\mu_4-3\mu_2^2)
+6f''\mu_3\mu_2''
-4\mu_3\left(f^{(4)}\mu_2-3f'''\mu_2'\right)
\Bigr]
\\
&+f'^3
\Bigl[
3f^{(5)}\mu_2^3
+18\mu_2^2
\left(
f^{(4)}\mu_2'
+f'''\mu_2''
\right)
\\
&\qquad
+2f''\mu_2
\left(
26f'''\mu_3
+3\mu_2\mu_2'''
\right)
-48f''^2\mu_3\mu_2'
\Bigr]
\\
&-f'^2\mu_2
\Bigl[
28f'''^2\mu_2^2
+216f''f'''\mu_2'\mu_2
\\
&\qquad
+3f''
\left(
15f^{(4)}\mu_2^2
+f''
\left(
32f''\mu_3
-\mu_2'^2
+26\mu_2\mu_2''
\right)
\right)
\Bigr]
\\
&+3f''^2f'\mu_2^2
\left(
101f'''\mu_2
+128f''\mu_2'
\right)
\\
&-297f''^4\mu_2^3
-24f'^8\mu_2\tilde b_2'
\Bigr).
\end{aligned}
\end{equation}
 
The coefficient identifies protocols for which the full \(1/\nu^2\) correction is cancelled:
\begin{equation}
 B_M=0.
\end{equation}
In these scenarios, the leading finite measurement correction is postponed to \(1/\nu^3\).
 
\section{Density-matrix first-order optimality}\label{sm:density_matrix_benchmark}
 
Let \(X=M-f(\theta)\mathbb I\), with \(f(\theta)=\Tr[\rho_\theta M]\), and let \(L_\theta\) be the SLD,
\begin{equation}
 \partial_\theta\rho_\theta
 =
 \frac{1}{2}\left(L_\theta\rho_\theta+\rho_\theta L_\theta\right),
 \qquad
 F_Q=\Tr[\rho_\theta L_\theta^2].
\end{equation}
Since \(\Tr[\partial_\theta\rho_\theta]=0\),
\begin{equation}
 f'(\theta)=\Tr[M\partial_\theta\rho_\theta]
 =\frac{1}{2}\Tr[\rho_\theta\{X,L_\theta\}].
\end{equation}
The covariance Cauchy--Schwarz inequality gives
\begin{equation}
 \left[\frac{1}{2}\Tr[\rho_\theta\{X,L_\theta\}]\right]^2
 \le
 \Tr[\rho_\theta X^2]\Tr[\rho_\theta L_\theta^2],
\end{equation}
and therefore \(\mu_2/(f')^2\ge 1/F_Q\). For full-rank models equality requires \(X=(f'/F_Q)L_\theta\). For rank-deficient models, components outside the support of \(\rho_\theta\) may be invisible to this leading condition but can enter higher derivatives of the calibration curve.
 
\section{The first-order optimal observables}\label{sm:optimal_measurement}
 
In this section, let us specify the pure unitary encoding \(|\psi_\theta\rangle=e^{-i\theta H}|\psi\rangle\). We identify \textit{all} first-order optimal observables (FOOs) in this scenario to minimize the leading term of MoM estimator, which coincides with QCRB. As stated in the main text, we have the following theorem
\begin{theorem}
    For pure states with unitary encoding $U(\theta)=\exp(-i\theta H)$, the SLD equation, up to a scale and shift, identifies all observables $M_{\rm opt}$ that reach QCRB in the first-order asymptotic regime with MoM.
\end{theorem}
 
In preparation, we introduce the \textit{moment picture} for convenience: based on every non-trivial encoding $U(\theta)=\exp(-i\theta H)$, we apply iteratively Gram-Schmidt orthogonalization to ${\rm span}(\ket\psi, \ket{\dot \psi}, \ket{\ddot \psi}, \ket{\dddot \psi},\cdots)$ to obtain the following orthonormal basis\begin{equation}\label{momentbasis}
    \begin{cases}
    \ket0:=\ket{\psi},\\
    \ket1:=\frac 1{\sqrt{G_Q^{(1)}}}(\ket{\dot \psi}-\braket{0|\dot\psi}\ket 0),\\
    \ket2:=\frac1{\sqrt{G_Q^{(2)}}}(\ket{\ddot \psi}-\braket{0|\ddot\psi}\ket 0-\braket{1|\ddot\psi}\ket 1),\\\ket3:=\frac 1{\sqrt{G_Q^{(3)}}}(\ket{\dddot \psi}-\braket{0|\dddot\psi}\ket 0-\braket{1|\dddot\psi}\ket 1-\braket{2|\dddot\psi}\ket 2),\\\cdots
    \end{cases}
\end{equation}where for simplicity we have defined the normalizing coefficients, the first few of which are listed below\begin{equation}
    \begin{aligned}
        &G^{(1)}_Q:=m_2\equiv \frac14 F_Q,\\&G^{(2)}_Q:=\frac{m_2(m_4-m_2^2)-m_3^2}{m_2},\\&G^{(3)}_Q:=\frac{m_3^4+m_2^2(m_4^2-m_2m_6)-m_3^2(3m_2m_4+m_6)-m_4^3+2m_3m_5(m_2^2+m_4)+m_2(m_4m_6-m_5^2)}{m_2(m_4-m_2^2)-m_3^2},
    \end{aligned}
\end{equation}where we denote the expectation value of Hamiltonian as $m:=\langle H\rangle$ and $n$-th central moment as\begin{equation}
    m_n:=\langle (H-\langle H\rangle )^n\rangle.
\end{equation}We can thus represent the Hamiltonian in an elegant form\begin{equation}\label{Hmoment}
    H=\begin{pmatrix}
        m&-i\sqrt{G_Q^{(1)}}&0&0&\cdots\\i\sqrt{G_Q^{(1)}}&m+r_1&-i\sqrt{\frac{G_Q^{(2)}}{G_Q^{(1)}}}&0&\cdots \\0&i\sqrt{\frac{G_Q^{(2)}}{G_Q^{(1)}}}&m+r_2&-i\sqrt{\frac{G_Q^{(3)}}{G_Q^{(2)}}}&\cdots \\0&0&i\sqrt{\frac{G_Q^{(3)}}{G_Q^{(2)}}}&m+r_3&\cdots \\\vdots&\vdots&\vdots&\vdots&\  
    \end{pmatrix},
\end{equation}where we have defined\begin{equation}
    r_1:=\frac{1}{G_Q^{(1)}}\cdot m_3,
\end{equation}\begin{equation}
    r_2:=\frac{1}{(G_Q^{(1)})^2G_Q^{(2)}}\cdot (m_3^3-2m_2m_3m_4+m_2^2m_5),
\end{equation}\begin{equation}
    \begin{aligned}
        r_3:=\frac{1}{G_Q^{(3)}}&\frac{1}{(G_Q^{(2)})^2}\frac{1}{(G_Q^{(1)})^2}\{m_7 m_2^6-2 \left(m_4 m_5+m_3 m_6\right) m_2^5+2 \left(-2 m_4 m_3^3+m_7 m_3^2-2 m_5^2 m_3+m_4^2 m_5\right) m_2^3\\&\left(3 m_5 m_3^2+3 m_4^2 m_3+2 m_5 m_6-2 m_4 m_7\right) m_2^4++m_3 \left(-m_4^4+3 m_3 m_5 m_4^2-m_3^2 \left(m_5^2+2 m_4 m_6\right)+m_3^3 m_7\right)\\&+m_2^2 \left(m_3^5-2 m_6 m_3^3+4 m_4 m_5 m_3^2-2 m_4^3 m_3+m_5^3-2 m_4 m_5 m_6+m_4^2 m_7\right)\\&+2 m_2 m_3 \left(m_6 m_4^2-\left(m_5^2+m_3 m_7\right) m_4+m_3 m_5 m_6\right)\}.
    \end{aligned}
\end{equation}
\begin{lemma}
    The Hamiltonian is of tri-diagonal form in~\cref{Hmoment}.
\end{lemma} 
 
\proof In other words, we need to prove\begin{equation}
    \bra{k+j}H\ket k=0, \forall k\ge 0, j\ge 2.
\end{equation}This can be directly observed since $H\ket k\in {\rm span}(\ket 0,\ket 1,\cdots, \ket{k+1})$, which is orthogonal to $\ket{k+l}$ due to the Gram-Schmidt procedures.\qed
 
Given its neat form of Hamiltonian, we can now identify the FOOs as well
\begin{lemma}
    In the case of a pure state, all observables that minimize the leading term of MoM estimator are those satisfying \begin{equation}\label{optcondition}
        \begin{cases}
            \bra 0 M\ket 1=\bra 1 M\ket 0\ne 0\\\bra 0 M\ket k=0,\forall k\ge 2
        \end{cases}.
    \end{equation}The minimum is actually QCRB\begin{equation}
        \min_{M}\frac{(\Delta M)^2_\theta}{(\partial_\theta \langle M\rangle_\theta)^2 }=\frac{1}{F_Q}.
    \end{equation}
\end{lemma}
 
\proof Let us assume the Hilbert space is of $d$-dimensional. With the moment basis $\{\ket0,\ket1,\ket2,\cdots,\ket{d-1}\}$, a generic observable $M$ can be formally put as $M=\sum_{i,j=0}^{d-1}M_{ij}\ket i\bra j$, while the Hamiltonian in ~\cref{Hmoment} reads\begin{equation}
    H=\sum_{k=0}^{d-2}i\sqrt{\frac{G_Q^{(k+1)}}{G_Q^{(k)}}}(\ket{k+1}\bra k-\ket k\bra{k+1})+\sum_{k=0}^{d-1}r_k\ket k\bra k,
\end{equation}where we have defined $G_Q^{(0)}:=1$ and $r_0:=0$. It can be verified that \begin{equation}
    \begin{cases}
        \Tr(\rho M^2)=\sum_{i=0}^{d-1}|M_{0i}|^2,\\\Tr(\rho M)=M_{00},\\\Tr(\rho M)=-i\sqrt{G_Q^{(1)}}(M_{01}+M_{10}).
    \end{cases}
\end{equation}Therefore, we derive the first-order method of moment in the following simple form\begin{equation}
    \frac{(\Delta M)^2_\theta}{(\partial_\theta \langle M\rangle_\theta)^2 }=\frac{\sum_{k=1}^{d-1}|M_{0k}|^2}{4G_Q^{(1)}({\rm Re\ }M_{01})^2},
\end{equation}which is obviously minimized by letting $M_{0k}=0,\forall k\ge 2$ and $M_{01}\in\mathbb R$. Noting $G_Q^{(1)}=\frac14 F_Q$, the minimum coincides with QCRB.
\qed
 
Let us recall the symmetric logarithmic derivative (SLD) operator~\cite{braunstein1994}, which is introduced via the first-order differential of quantum state\begin{equation}\label{slddef}
    \partial_\theta \rho_\theta=\frac12(L\rho_\theta+\rho_\theta L).
\end{equation}The solution of~\cref{slddef} saturates QCRB with the first-order method of moments~\cite{gessner2019}, which provides in fact only a \textit{sufficient} condition, since it was introduced in an estimator-free context. When the state is \textit{pure}, the analytical form of SLD \textit{often} reads 
\begin{equation}\label{SLDpure}
    L=2(\ket{\psi}\bra{\dot \psi}+\ket{\dot \psi}\bra\psi),
\end{equation}which is of rank-2 for non-trivial encoding as a result of $\partial_\theta \rho=\ket{\psi}\bra{\dot \psi}+\ket{\dot \psi}\bra\psi$. 
Note that the SLD form in~\cref{SLDpure} is derived by projecting~\cref{slddef} onto the span of $\ket{\psi},\ket{\dot\psi}$, neglecting the other dimensions of the SLD. Using the same picture, we can identify the complete solutions to \cref{slddef}
\begin{lemma}
    In the case of a pure state, all the solutions $L$ to the SLD equation~\cref{slddef} are those satisfying \begin{equation}
        \begin{cases}
            \bra 0 L\ket 1=\bra 1 L\ket 0=\frac12\sqrt{F_Q}\\\bra 0 L\ket k=0,\forall k\ne1
        \end{cases}.
    \end{equation}
\end{lemma}
 
With the above lemmas, we see that the solutions, up to a global multiplier and identity, are essentially the same as the requirements for optimal measurements so that by substituting $L\rightarrow2\lambda M_{\rm opt}+\frac12\mu \mathbb {I
}, \forall \lambda,\mu \in\mathbb R$ into~\cref{slddef}, we have proved the relating arguments in the main text.
 
\section{Qubit case}\label{sm:qubit}
In this section, we showcase our theory based on the most studied qubit system. Using the basis in~\cref{momentbasis}, the state is simply 
\begin{equation}
    \rho=\begin{pmatrix}
        1&0\\0&0
    \end{pmatrix}.
\end{equation}
The Hamiltonian can be reformulated as\begin{equation}
    H=\begin{pmatrix}
        m&-i\sqrt{m_2}\\i\sqrt{m_2}&m+\frac{m_3}{m_2}
    \end{pmatrix}.
\end{equation}Note that on can always express the higher-order moments of $H$ in 2-dimensional system with the first three orders $m,m_2,m_3$. The FOOs read
\begin{equation}
 M=\begin{pmatrix}0&1\\1&\lambda\end{pmatrix},\forall \lambda\in \mathbb R.
\end{equation}
 
The calibrating function of the qubit is
\begin{equation}
f(\theta)=\frac{2 \sqrt{m_2^3 \left(4 m_2^3+m_3^2\right)} \sin \left(\frac{\theta  \sqrt{4 m_2^3+m_3^2}}{m_2}\right)-2 \lambda  m_2^3 \left(\cos \left(\frac{\theta  \sqrt{4 m_2^3+m_3^2}}{m_2}\right)-1\right)}{4 m_2^3+m_3^2}.
\end{equation}
Based on Eq.~\eqref{mse2} and Eq.\eqref{third_explicit_bc} respectively, we derive the corrections to the original MoM estimator (true value $\theta =0$)
\begin{equation}
 \MSE(\hat\theta_0)= \frac{1}{4m_2\nu} +\frac{\left(7 \lambda ^2+16\right) m_2^3+4 m_3^2}{64 m_2^4\nu^2}
+O(\nu^{-3}),
\end{equation}
and the biased-corrected estimator
\begin{equation}
 \MSE(\hat\theta_{\rm bc})=\frac{1}{4m_2\nu}+\frac{\lambda^2}{32m_2\nu^2} +O(\nu^{-3}).
\end{equation}With a centered choice of \(\lambda=0\), the \(1/\nu^2\) correction of \(\hat \theta_0\) is minimized and \(\hat \theta_{\rm bc}\) cancelled. In this case, we expand to the next order as\begin{equation}
 \MSE(\hat\theta_{\rm bc})=\frac{1}{4m_2\nu} +\frac{(4m_2^3+m_3^2)^2}{384m_2^7\nu^3}
 +O(\nu^{-4}).
\end{equation}
 
\begin{figure*}[t]
\centering
\begin{minipage}{0.48\textwidth}
\centering
\maybegraphics[width=\linewidth]{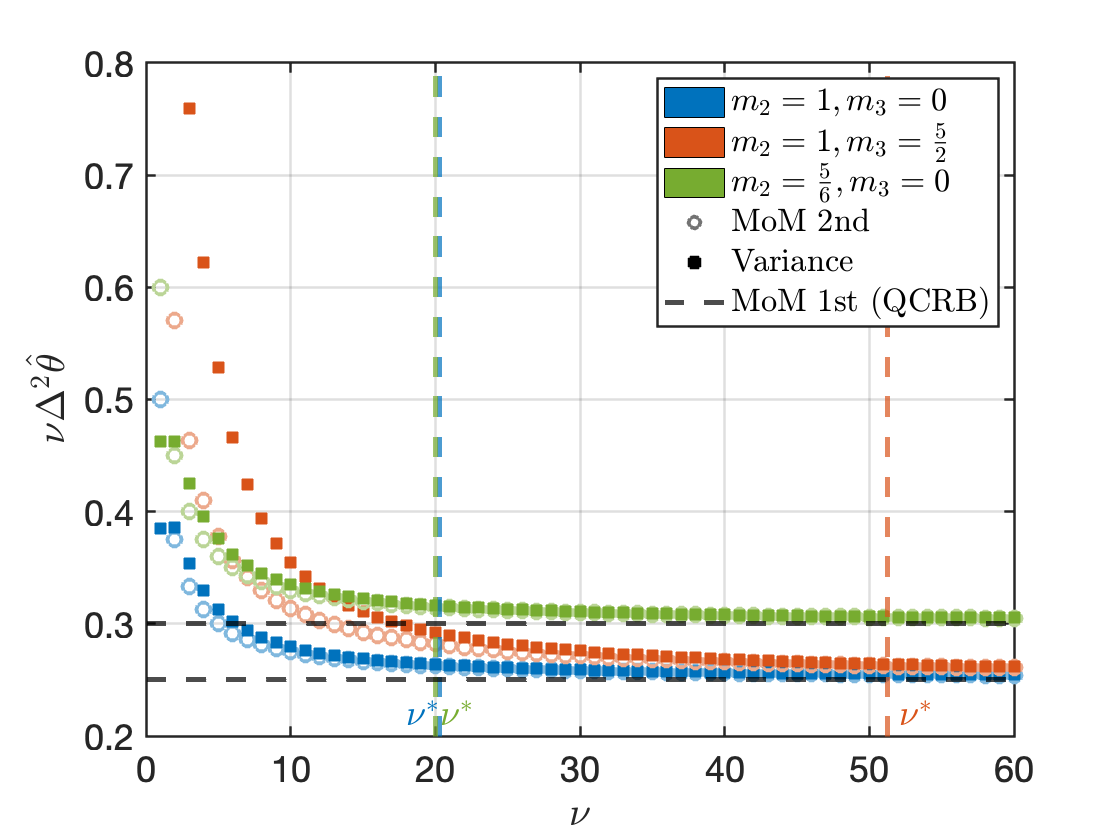}
\end{minipage}\hfill
\begin{minipage}{0.48\textwidth}
\centering
\maybegraphics[width=\linewidth]{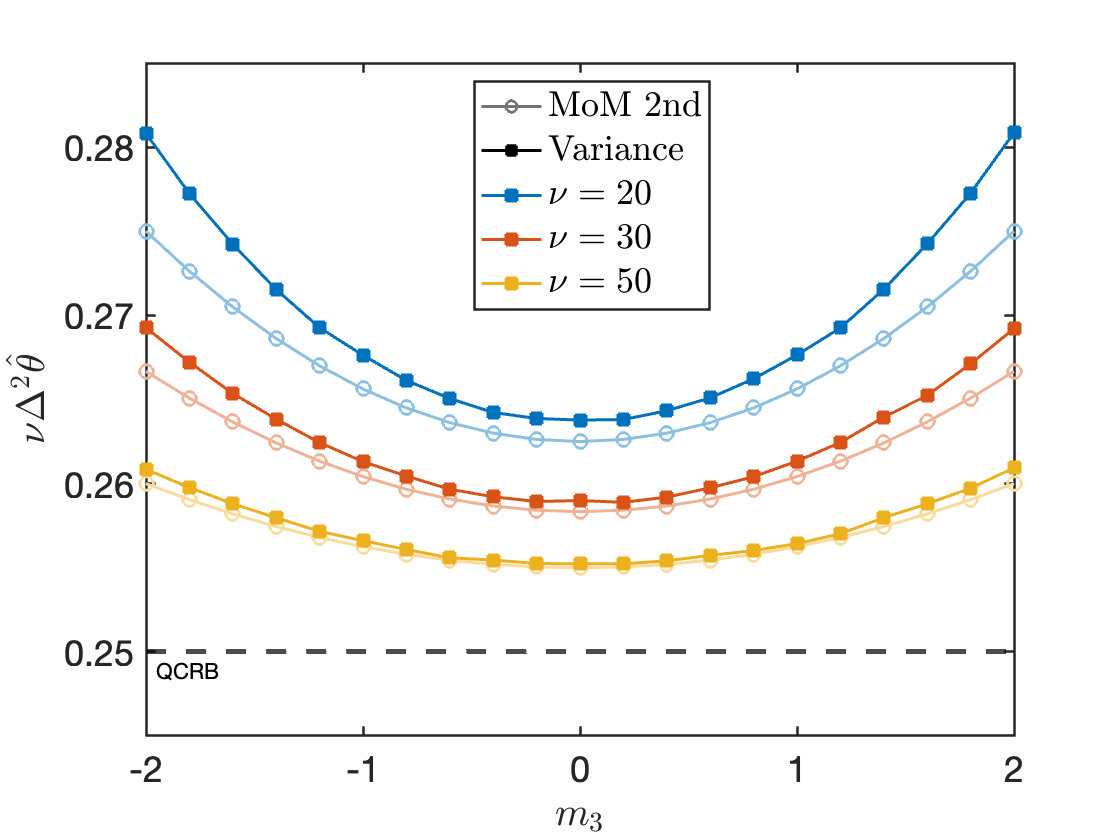}
\end{minipage}
\caption{Effective-qubit finite-measurement behavior of the usual MoM estimator $\hat \theta_0$. The panels reproduce the original Monte Carlo diagnostics for the uncorrected MoM estimator and the corresponding second-order variance expansion. In the bias-corrected formulation, these plots should be interpreted as diagnostics of the finite-measurement regime and of the role of the generator moments, rather than as the final corrected-MSE expansion. The corrected theory removes the spurious estimator-bias contribution and shifts the centered case to the \(1/\nu^3\) law in Eq.~\eqref{eq:qubit_third_order_bc}.}
\label{fig:qubit_old_main}
\end{figure*}
 
We consider the general Heisenberg-limited scenarios. Given the Hamiltonian, the probe state is cat-like state superposed by the eigenstates corresponding to minimal and maximal eigenvalues. For cat-like or Heisenberg-limited probes supported on the maximal and minimal eigenstates of the Hamiltonian, we find that
\begin{equation}
    m_k=\begin{cases}
        \frac{1}{2^k}(\lambda_M-\lambda_m)^k&\text{for even }k\\
        0&\text{for odd }k
    \end{cases}
\end{equation}
Therefore, the correction of the bias-corrected estimator reduces to
\begin{equation}
 \MSE(\hat\theta_{\rm bc})
 =
 \frac{1}{4m_2\nu}
 +\frac{1}{24m_2\nu^3}
 +O(\nu^{-4}).
\end{equation}
 
\section{Qudit case}\label{sm:qudit}
In this section, we consider the general case of the \textit{effective} high-dimensional system, with the Hamiltonian~\cref{Hmoment}. We denote by \({\rm Re}\,M_{23}:=\alpha\) a generic FOO. With again the centered choice of vanishing diagonal terms, we find that for the usual MoM estimator
\begin{equation}\label{varqudit2s}
\MSE(\hat \theta_0)=\frac1{\nu}\frac1{4m_2}+\frac{1}{\nu^2}\frac{D_\alpha}{16m_2^3},
\end{equation}
while for the bias-corrected one
\begin{equation}
 \MSE(\hat\theta_{\rm bc})
 =
 \frac{1}{4m_2\nu}
 +\frac{D_\alpha^2}{384m_2^5\nu^3}
 +O(\nu^{-4}),
\end{equation}where we call\begin{equation}
    D_\alpha= 3m_2^2+m_4-3\alpha m_2\sqrt{G_Q^{(2)}}.
\end{equation}
The choice
\begin{equation}
 \alpha_{\rm opt}= \frac{3m_2^2+m_4}{3m_2\sqrt{G_Q^{(2)}}}
\end{equation}
cancels the displayed \(1/\nu^3\) term within the same local perturbative expansion. 
 
We provide with another concrete example to demonstrate the metrological advantage of high-dimensional system. We construct the qutrit Hamiltonian\begin{equation}\label{exampletrit}
    H_t=\begin{pmatrix}
       0&-i&0 \\i&0&-\sqrt2 i\\0&\sqrt2 i&0
    \end{pmatrix},
\end{equation}
where we choose \(m_2=1\), \(m_3=0\), \(m_4=3\), and therefore \(G_Q^{(2)}=2\). We adopt the optimal measurement as\begin{equation}
    M_t=\begin{pmatrix}
       0&1&0 \\1&0&\alpha\\0&\alpha&0
    \end{pmatrix}.
\end{equation}
For this model, the calibrating function is
\begin{equation}
f(\theta)=-\frac{2 \sin \left(\sqrt{3} \theta \right) \left(\left(\sqrt{2} \alpha -1\right) \cos \left(\sqrt{3} \theta \right)-\sqrt{2} \alpha -2\right)}{3 \sqrt{3}}.    
\end{equation}
A direct symbolic expansion of the bias-corrected estimator gives (true value $\theta =0$)
\begin{equation}
 \MSE(\hat\theta_{\rm bc})=\frac{1}{4\nu}+\frac{3(\alpha^2-2\sqrt2\alpha+2)}{64\nu^3}+O(\nu^{-4}).
\end{equation}
 
\begin{figure}[t]
\centering
\maybegraphics[width=0.5\columnwidth]{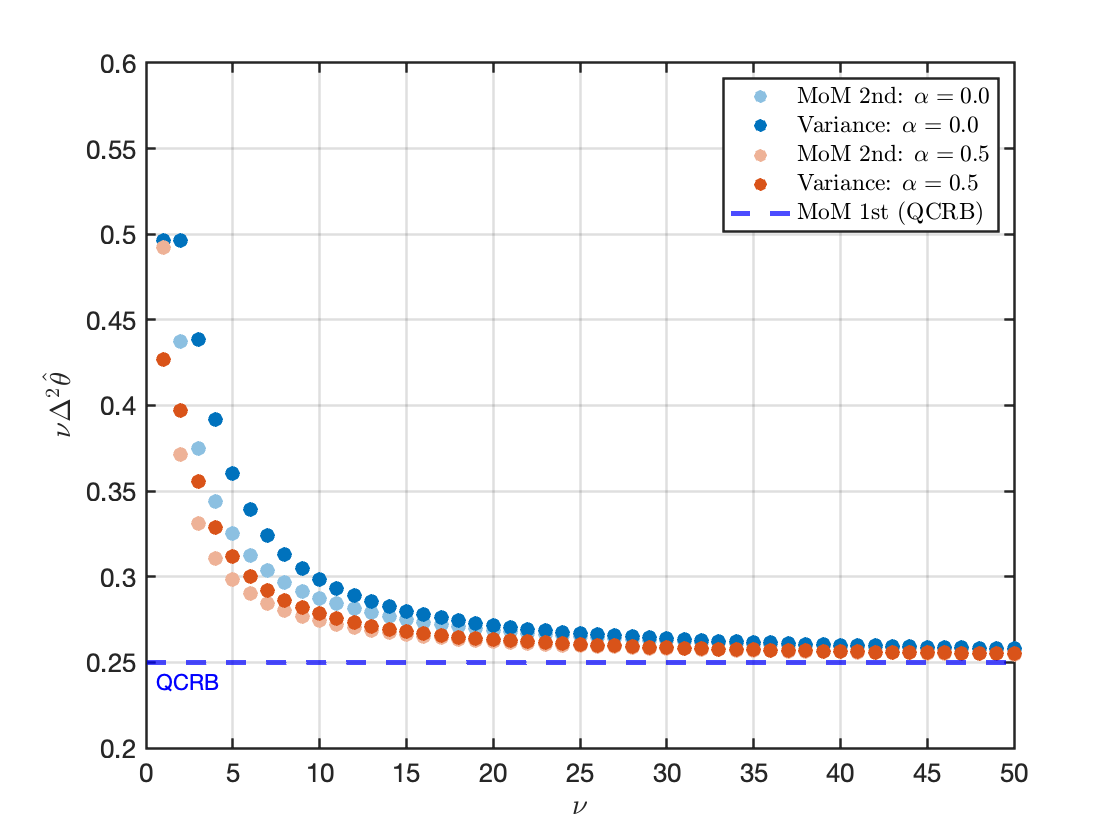}
\caption{Qutrit Monte Carlo benchmark from the original analysis of the usual MoM estimator $\hat \theta_0$. The two values of \(\alpha\) illustrate that measurement components invisible at leading order can produce different finite-measurement convergence. In the corrected theory this effect is captured by the \(\alpha\)-dependent \(1/\nu^3\) coefficient in Eq.~\eqref{eq:qutrit_check}.}
\label{fig:qutrit_benchmark_main}
\end{figure}
 
\begin{figure}[t]
\centering
\maybegraphics[width=0.5\columnwidth]{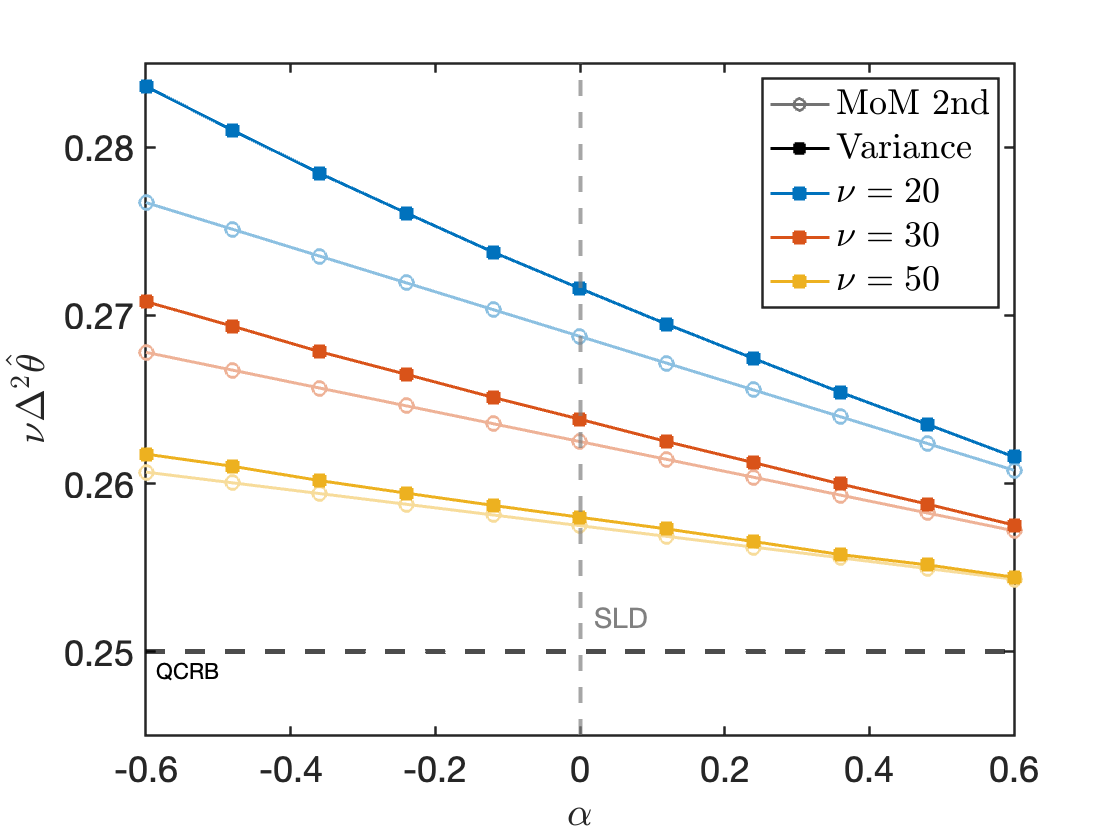}
\caption{Qudit finite-measurement behavior with a higher-rank measurement component of the usual MoM estimator $\hat \theta_0$. The figure reproduces the original Monte Carlo diagnostic in which the parameter \(\alpha\) changes the finite-sample convergence at fixed leading sensitivity. In the corrected formulation, the same component controls the curvature coefficient \(D_\alpha\) and can reduce or cancel the displayed \(1/\nu^3\) correction.}
\label{fig:qudit_alpha_main}
\end{figure}
 
In comparison, we now study the maximum likelihood (ML) estimator based on the measurement $M$ to show the influence of high-dimensional metrology. By the measurement axiom of quantum mechanics, $M$ yields three outcomes at $\theta=0$ (labelled with $1,2,3$) with probabilities\begin{equation}
    p_1=\frac{\alpha^2}{1+\alpha^2},\quad
    p_2=\frac{1}{2(1+\alpha^2)},\quad
    p_3=\frac{1}{2(1+\alpha^2)}.
\end{equation}We give here the first three derivatives of them as well for further use\begin{equation}
    \begin{cases}
        \partial_\theta p_1=0\\\partial^2_\theta p_1=-\frac{2\alpha (\alpha G_Q^{(1)}+\sqrt{G_Q^{(2)}})}{1+\alpha^2}\\\partial^3_\theta p_1=0
    \end{cases},\ \begin{cases}
        \partial_\theta p_2=-\sqrt{\frac{G_Q^{(1)}}{1+\alpha^2}}\\\partial^2_\theta p_2=\frac{\alpha (\alpha G_Q^{(1)}+\sqrt{G_Q^{(2)}})}{1+\alpha^2}\\\partial^3_\theta p_2=\frac{4(G_Q^{(1)})^2+G_Q^{(2)}+G_Q^{(2)}(-3\alpha \sqrt{G_Q^{(2)}}+r_1^2)}{\sqrt{G_Q^{(1)}(1+\alpha^2)}}
    \end{cases},\ \begin{cases}
        \partial_\theta p_3=\sqrt{\frac{G_Q^{(1)}}{1+\alpha^2}}\\\partial^2_\theta p_3=\frac{\alpha (\alpha G_Q^{(1)}+\sqrt{G_Q^{(2)}})}{1+\alpha^2}\\\partial^3_\theta p_3=-\frac{4(G_Q^{(1)})^2+G_Q^{(2)}+G_Q^{(1)}(-3\alpha \sqrt{G_Q^{(2)}}+r_1^2)}{\sqrt{G_Q^{(1)}(1+\alpha^2)}}
    \end{cases}.
\end{equation}
 
The original paper of Rao~\cite{rao1965} developed classical analytical rectifications of a class of estimators including the ML estimator. Given the second-order correction\begin{equation}\label{rao}
    \MSE (\hat\theta _{\rm ML})=\frac1\nu\frac1{F(\theta)}+\frac1
    {\nu^2}[\Gamma(\theta)+b^2(\theta)+2\frac{b'(\theta)}{F(\theta)}]+O(\frac1{\nu^3}),
\end{equation}where the classical Fisher information and other quantities are defined as\begin{equation}
    F(\theta):=\sum_{\mu=1}^K \frac{(\partial _\theta p_\mu)^2}{p_\mu},
\end{equation}\begin{equation}
    b(\theta):=\lim_{\nu\rightarrow\infty} \nu\langle\hat \theta_{ML}-\theta \rangle=-\frac{\gamma_{11}}{2F(\theta)^2},
\end{equation}\begin{equation}
    \Gamma(\theta):=-\frac1{F(\theta)}+\frac{\gamma_{02}-2\gamma_{21}+\gamma_{40}}{F^3(\theta)}+\frac{\frac12\gamma_{11}^2-[\gamma_{11}-\gamma_{30}]^2}{F^4(\theta)},
\end{equation}with\begin{equation}
    \gamma_{rs}:=\sum_{\mu=1}^K p_{\mu}{(\frac{\partial _\theta p_\mu}{p_\mu})^r}{(\frac{\partial^2 _\theta p_\mu}{p_\mu})^s}.
\end{equation}Inserting the derivatives of the PDFs, we finally obtain the second-order correction to ML estimator in our qutrit setting as\begin{equation}
    \MSE (\hat\theta _{\rm ML})=\frac1{\nu}\frac1{4m_2}+\frac{1}{\nu^2}\frac{3(1+2\alpha^2)m_2^2+m_4-5\alpha\cdot G_Q^{(1)} \sqrt{G_Q^{(2)}}}{16m_2^3}+O(\frac1{\nu^3}).
\end{equation}
 
\end{document}